\documentclass[11pt]{article}

\usepackage[margin=1in]{geometry}
\usepackage{amsmath,amssymb,amsthm,mathtools}
\usepackage{microtype}
\usepackage{booktabs}
\usepackage{enumitem}
\usepackage{xcolor}

\usepackage[lined,boxed,ruled,norelsize,algo2e,linesnumbered,noend]{algorithm2e}
\usepackage[hypertexnames=true]{hyperref}
\usepackage{aliascnt}
\hypersetup{
  colorlinks=true,
  linkcolor=blue!60!black,
  citecolor=blue!60!black,
  urlcolor=blue!60!black
}
\usepackage[nameinlink,capitalize,noabbrev]{cleveref}

\usepackage{tikz}
\usetikzlibrary{arrows.meta}

\newtheorem{theorem}{Theorem}[section]
\newaliascnt{lemma}{theorem}
\newtheorem{lemma}[lemma]{Lemma}
\aliascntresetthe{lemma}
\newaliascnt{corollary}{theorem}

\aliascntresetthe{corollary}
\newaliascnt{proposition}{theorem}
\newtheorem{proposition}[proposition]{Proposition}
\aliascntresetthe{proposition}
\newaliascnt{definition}{theorem}
\newtheorem{definition}[definition]{Definition}
\aliascntresetthe{definition}
\newaliascnt{remark}{theorem}

\aliascntresetthe{remark}
\newaliascnt{conjecture}{theorem}
\newtheorem{conjecture}[conjecture]{Conjecture}
\aliascntresetthe{conjecture}
\newaliascnt{fact}{theorem}
\newtheorem{fact}[fact]{Fact}
\aliascntresetthe{fact}

\crefname{theorem}{Theorem}{Theorems}
\crefname{lemma}{Lemma}{Lemmas}
\crefname{corollary}{Corollary}{Corollaries}
\crefname{proposition}{Proposition}{Propositions}
\crefname{definition}{Definition}{Definitions}
\crefname{remark}{Remark}{Remarks}
\crefname{conjecture}{Conjecture}{Conjectures}
\crefname{fact}{Fact}{Facts}

\renewcommand{\epsilon}{\varepsilon}
\newcommand{\eps}{\epsilon}
\newcommand{\defeq}{\mathrel{:=}}

\newcommand{\E}{\mathbb{E}}

\newcommand{\tO}{\widetilde O}

\newcommand{\Par}[1]{\left(#1\right)}
\newcommand{\Abs}[1]{\left|#1\right|}
\newcommand{\norm}[1]{\left\lVert#1\right\rVert}

\newcommand{\ceil}[1]{\left\lceil#1\right\rceil}

\renewcommand{\AA}{\boldsymbol{\mathrm A}}
\newcommand{\DD}{\boldsymbol{\mathrm D}}
\newcommand{\II}{\boldsymbol{\mathrm I}}
\newcommand{\LL}{\boldsymbol{\mathrm L}}
\newcommand{\MM}{\boldsymbol{\mathrm M}}
\newcommand{\NN}{\boldsymbol{\mathrm N}}
\newcommand{\PP}{\boldsymbol{\mathrm P}}
\newcommand{\bb}{\boldsymbol{\mathrm b}}
\newcommand{\dd}{\boldsymbol{\mathrm d}}
\newcommand{\ee}{\boldsymbol{\mathrm e}}
\newcommand{\hh}{\boldsymbol{\mathrm h}}
\newcommand{\ppi}{\boldsymbol{\mathrm \pi}}

\newcommand{\ssig}{\boldsymbol{\mathrm \sigma}}

\newcommand{\vv}{\boldsymbol{\mathrm v}}
\newcommand{\ww}{\boldsymbol{\mathrm w}}
\newcommand{\xx}{\boldsymbol{\mathrm x}}

\newcommand{\vone}{\boldsymbol{\mathrm 1}}
\newcommand{\vzero}{\boldsymbol{\mathrm 0}}

\newcommand{\ER}{\operatorname{ER}}

\DontPrintSemicolon
\SetFuncSty{textsc}
\SetKwProg{Pro}{Procedure}{}{}
\SetKwInOut{Input}{Input}
\SetKwInOut{Output}{Output}
\SetKwFunction{Initialize}{Initialize}
\SetKwFunction{Delete}{Delete}
\SetKwFunction{Insert}{Insert}
\SetKwFunction{Update}{Update}
\SetKwFunction{Query}{Query}
\SetKwFunction{Median}{Median}
\SetKw{OMvR}{OMvR}

\title{Fully Dynamic Random-Walk Quantities on Undirected Expanders and
$\kappa$-Uniform Directed Graphs}
\author{
\begin{tabular}{c c}
Richard Peng & Sushant Sachdeva\\
yangp@cs.cmu.edu & sachdeva@cs.toronto.edu\\
Carnegie Mellon University & University of Toronto
\\[2ex]
Junzhao Yang & Yibin Zhao\\
junzhaoy@andrew.cmu.edu & ybzhao@cs.toronto.edu\\
Carnegie Mellon University & University of Toronto
\end{tabular}
}
\date{}

\begin{document}
\maketitle

\begin{abstract}
We study fully dynamic data structures for random-walk quantities on
undirected expanders and fast-mixing directed graphs.
For an unweighted
graph whose normalized Laplacian has spectral gap at least $\gamma$, we give a
$(1\pm\eps)$ effective-resistance data structure with expected update time
$O(\gamma^{-4}\eps^{-2}\log n\log^4(n/\eps))$ and expected $O(1)$ query time.

For directed graphs, somewhat surprisingly,
the OMv conjecture precludes a fast data structure
based on rapid mixing conditions alone.
To address this, we introduce the notion of $\kappa$-uniform directed graphs, which requires the stationary distribution to $\kappa$-approximate the uniform distribution,
and show that in this setting,
random walk based data structures with $\textnormal{poly}(\kappa, \log{n}, \epsilon^{-1})$  update/query times are possible.
\end{abstract}

\section{Introduction}

Expanders and graph partitioning are highly useful primitives in designing efficient graph algorithms~\cite{SaranurakW19,GoranciRST21}.
An approach starting from the celebrated R{\"a}cke trees~\cite{R02} provided ways to solve a problem on general graphs by solving a sequence of problems on expanders~\cite{C14},
while stitching together solutions at their boundaries.
Such approaches have been behind many breakthroughs in both static and dynamic graph algorithms~\cite{SpielmanT04,bernstein2020deterministic,ChenKLMP24,EHL25,EHL26}.

One significant barrier facing the expander paradigm of graph algorithm design is that the algorithm, or data structure,
must work when the underlying graph is an expander.
While this is often the case for static problems,
and provided much of the motivation of these expander-based approaches,
the status of such primitives for dynamic graph problems,
especially fully dynamic graphs,
is surprisingly unexplored.
In particular, it is open whether one can maintain, in subpolynomial time,
$(1 + \epsilon)$-approximate $s$-$t$ min-cut or shortest path
in fully dynamic undirected graphs that always remain expanders~\cite{GoranciRST21,HLS24}.

In this paper we study data structures for maintaining random walk related quantities on fully dynamic expanders.
We give positive results on undirected graphs, as well as show the limitations of these methods on directed graphs.
We build upon the static data structure by Li and Sachdeva \cite{LiS23}, which showed that short random walks suffice to answer approximate effective resistance queries in essentially constant time.
By maintaining dynamic versions of these walks on an unweighted $\gamma$ spectral expander, we show that a data structure supporting approximate all-pairs effective resistance queries can be maintained in only
$ \widetilde{O}\left(\gamma^{-c}\right) $
update time, for a constant $c$.
This gives a substantial improvement over what is known for general dynamic graphs, for which the best runtime is about $m^{2/3}$\cite{ChenGHPS20}.

The random walk interpretation of commute time and the positive result for undirected expanders naturally suggest that similarly efficient dynamic data structures might exist for directed graphs whose random walks also mix rapidly.
Here however an obstacle appears: rapid mixing alone places little restriction on the stationary distribution of a directed random walk, which may be highly non-uniform.
By leveraging this fact, we give a conditional lower bound based on the OMv conjecture~\cite{HenzingerKNS15} that no data structure can simultaneously have update and query ($3/2$-approx) time $O\left(n^{1/2-\delta}\right)$,
even when the graph has mixing time $O(\log{n})$, and all outdegrees are constant.

This hardness result in turn suggests a natural additional condition for rapidly mixing directed graphs:
the stationary distribution is also close to uniform.
We incorporate this by defining $\kappa$-uniform graphs in~\cref{def:kappa-uniform}, and show that this setting enables the rapid query of stationary distributions.
This in turn allows random walk based approaches on the Eulerian graph obtained by rescaling,
and thus efficient commute time queries.
We believe this definition that combines both rapid mixing and uniformity of the stationary distribution provides a new setting for studying both dynamic as well as static graph algorithms.

\subsection{Overview}
For undirected expanders, our efficient dynamic data structure for computing effective resistance, and equivalently commute time, stems from the observation
in the static graph case by Li and Sachdeva \cite{LiS23} that the hitting statistics computed by a collection of short, polylogarithmic-length random walks starting from each vertex suffice to generate a data structure for answering high-accuracy effective-resistance queries on expanders in $O(1)$ time. 

The core of our efficient dynamic data structure is a fully dynamic algorithm for maintaining such uniform short random walks under edge insertions and deletions in unweighted undirected graphs (not necessarily expanders). 
In particular, we observe that only roughly a $1/d_v$ fraction of the walks incident to an affected vertex $v$ need to be resampled. 
Crucially, we also discover that the congestion of each vertex (i.e., the number of visits to a vertex by all random walks) scales at most linearly with its degree $d_v$.
Given that the total number of random-walk steps is $\tO(n)$, we can guarantee an update time of only $\tO(1)$.

For directed graphs, we discover that the fast-mixing property alone cannot give a similarly efficient algorithm. 
Many Markov-chain quantities on directed graphs (e.g., hitting time and commute time) require high-accuracy
estimates of the stationary distribution, which is a computationally challenging task for dynamic directed graphs, as we establish through a conditional hardness result based on OMv.

For the hardness result, we develop a reduction from an instance of matrix--vector multiplication of dimension $n$ to a fast-mixing directed graph of size $O(n^2)$. 
The main component of the graph is a standard construction that encodes each entry of the vector $\MM\vv$.  In particular, within this component, the entry $[\MM\vv]_i=1$ if and only if there is a directed path from a unique source $s$ to a vertex $r_i$ corresponding to the $i$th entry.
We show that the existence or absence of such a path is sufficient to create a difference in the stationary distribution at vertex $r_i$, allowing us to distinguish the value of $[\MM\vv]_i$. 
The main component of the graph consists of two connection layers, one for $\MM$ and one for $\vv$, where each edge in the layers maps one-to-one to an entry of $\MM$ or $\vv$. 
Changing $\vv$ (or $\MM$ between instances) corresponds directly to changing these connections
and, thus, the existence of paths from $s$ to the vertices $r_i$.

To ensure that the graph is strongly connected and $\kappa$-mixing, we also include two auxiliary components. 
The first component gives compact forward connections from $s$ to all vertices in the main component, ensuring
reachability. 
However, this gives a connection from $s$ to an $r_i$ that previously did not have a path. 
To ensure these connections do not interfere with the main paths, we give them a long chain of connections at the beginning, with many reset edges back to $s$ distributed along it, so that all of these connections are assigned extremely low probabilities.

Finally, all other vertices also have such reset edges back to $s$, so that at each non-source vertex, the probability of returning to $s$ is at least $1/4$.
This creates an extremely skewed stationary distribution in which $s$ receives a large mass from everywhere else. 
Such a construction also gives fast mixing of the graph, which can be proved by a standard random walk coupling argument: two walks starting from different vertices can be coupled to meet at $s$ with probability at least $1/4$ at each step. 
Consequently, our graph is always $O(\log n)$-mixing for any possible $\MM$ and $\vv$.

Each update of $\vv$ incurs only $O(n)$ updates and $O(n)$ stationary entry queries on the graph, for a total of $O(n^2)$ updates and queries throughout
the instance, on a graph of size $N=O(n^2)$.  This gives a roughly square-root
update- and query-time lower bound for this problem.

Finally, for our $\kappa$-uniform directed graphs, we use the uniformity of the stationary distributions to avoid the problem of such skewed resets. 
Efficient dynamic algorithms are once again possible for such graphs. 
For both stationary distribution and commute time approximation, we give observations analogous to those for undirected expanders: (1) there are data structures for efficiently querying these quantities that can be computed from short random walks; and (2) such short random walks incur only $\tO(\kappa^c)$ congestion at each vertex. 
Crucially, (2) is only true from the uniformity property; $\kappa$-mixing alone cannot guarantee it. 
These observations allow us to reuse the same dynamic short random walk algorithm from the undirected expander case with minimal changes and obtain efficient dynamic algorithms.

\section{Preliminaries}

Let $G=(V,E)$ be an unweighted, undirected (multi)graph with adjacency matrix $\AA$ and degree matrix $\DD$. 
Its Laplacian matrix is $\LL = \DD - \AA$ and the normalized Laplacian is $\NN=\II-\DD^{-1/2}\AA\DD^{-1/2}$.

For a directed (multi)graph $G=(V,E)$, we again use $\AA$ to denote its adjacency matrix while $\DD$ is the diagonal matrix of out-degrees.
The lazy random walk transition matrix is defined by $\PP = \frac{1}{2}(\II+\AA\DD^{-1})$.

A connected undirected graph is a $\gamma$-spectral expander if its spectral gap is at least $\gamma$.
\begin{definition}[Spectral expander]
\label{def:spectral-expander}
For $\gamma>0$, we call $G$ a $\gamma$-spectral expander if
$\lambda_2(\NN_G)\ge\gamma$.
\end{definition}

Expanders are naturally fast mixing.
\begin{proposition}[Expanders are fast mixing]
\label{prop:expander-mixing}
A $\gamma$-spectral expander is $O(\gamma^{-1}\log n)$-mixing.
\end{proposition}
\begin{proof}
$\lambda_2(\NN_G)\ge\gamma$ implies that
each nontrivial eigenvalue of $\PP$ has absolute value at most $1-\gamma/2$.
The convergence is with respect to the inverse degree.
Since the maximum degree in an unweighted graph is at most $n$, after a scaling, it suffices to take $\ell \ge C \gamma^{-1} \log n$ to make the entry-wise stationary approximation at most $n^{-10}$ as required.
\end{proof}

We formally define fast-mixing directed graphs.
\begin{definition}[$\kappa$-mixing]
\label{def:kappa-mixing}
For $\kappa \ge 1$, a directed graph $G$ is
$\kappa$-mixing\footnote{Common definitions of $\kappa$-mixing are equivalent up to logarithmic factors in the mixing time. We choose the exponent $10$ for convenience.} 
if it is strongly connected and, for the stationary distribution $\ppi$ of $\PP$,
\begin{equation}
    |(\PP^\kappa)_{vu}-\ppi_v|\le n^{-10}
    \qquad\text{for all }u,v\in V.
    \label{eq:directed-kappa-mixing}
\end{equation}
\end{definition}

For $\kappa$-uniform graphs, we additionally require the stationary distribution to be close to uniform.
\begin{definition}[$\kappa$-uniform]
\label{def:kappa-uniform}
A $\kappa$-mixing directed graph is $\kappa$-uniform if its stationary
distribution also satisfies
\begin{equation}
    \frac{1}{\kappa n}\le \ppi_v\le\frac{\kappa}{n}
    \qquad\text{for all }v\in V.
    \label{eq:directed-kappa-uniform}
\end{equation}
\end{definition}

We define hitting vectors conceptually by $\hh_v \defeq \sum_{k=0}^\infty \PP^k \ee_v$ for a vertex $v$, which diverges.
For distinct $u,v$, we denote the difference of their hitting vectors by
\[
    \hh_{u,v} \defeq \sum_{k=0}^{\infty}\PP^k(\ee_u - \ee_v) = \hh_u - \hh_v 
\]
which is typically well-defined.

We also define the truncated variants of these hitting vectors at length $\ell$ by $\hh^{(\ell)}_v \defeq \sum_{k=0}^{\ell-1} \PP^k \ee_v$, $\hh^{(\ell)}_{u,v} \defeq \hh^{(\ell)}_u - \hh^{(\ell)}_v$.
Similarly, the exact length variants are
$\hh^{(=\ell)}_v \defeq \PP^\ell \ee_v$, $\hh^{(=\ell)}_{u,v} \defeq \hh^{(=\ell)}_u - \hh^{(=\ell)}_v$.

We denote by $H_{u}(v)$ the hitting time of a random walk from $u$ to $v$, i.e. the expected number of steps until a (non-lazy) random walk starting from $u$ first visits $v$.
The commute time $C_{u,v}$ is the expected time that a random walk starting from one of the vertices visits the other vertex and returns.
In both undirected and directed cases, $C_{u,v} = H_{u,v}+H_{v,u}$.

On undirected graphs, the commute time between a pair of vertices is equivalent to the effective resistance between them.
\begin{fact}[Commute time and effective resistance]
\label{fact:undirected-commute-er}
Let $G=(V,E)$ be a connected unweighted undirected graph with $m=|E|$. Then, the commute time satisfies $C_G(u,v) = 2m \ER_G(u,v)$.
\end{fact}

The hitting and commute times of random walks can be quantified using the hitting vectors and the stationary distribution of the random walks.
\begin{fact}
\[
H_u(v) = -\frac{\hh_{u,v}(v)}{2\ppi_v}, \quad
H_v(u) = \frac{\hh_{u,v}(u)}{2\ppi_u}, \quad
C_{uv} = \frac{\hh_{u,v}(u)}{2\ppi_u} - \frac{\hh_{u,v}(v)}{2\ppi_v}.
\]
\end{fact}

For our conditional hardness result, we invoke the OMv conjecture.
This is a standard assumption for fine-grained complexity questions on dynamic data structures.
\begin{conjecture}[OMv conjecture~\cite{HenzingerKNS15}]
\label{conj:omv}
For every constant $\delta>0$, no randomized algorithm solves Online Boolean Matrix-Vector Multiplication (OMv) for a boolean matrix of size $n \times n$ and vectors of size $n$ in total time $O(n^{3-\delta})$ with probability at least $2/3$.
\end{conjecture}

\section{Fully Dynamic Effective Resistances on Expanders}
\label{sec:undir}

Li and Sachdeva \cite{LiS23} demonstrated that for expanders, a sparse data structure can be efficiently computed from which one can query the effective resistance between any pair of vertices.
Specifically, for a $\gamma$-spectral expander, Theorem 3.1 of
\cite{LiS23} shows that it suffices to sample for each vertex $(\gamma^{-1} \eps^{-2} \log n \log \frac{n}{\eps})$ random walks for each length $0 \le k < O(\gamma^{-1} \log \frac{n}{\eps})$
to support a $(1+\eps)$-approximate effective resistance query on a pair of vertices in expected time $O(1)$.

In \Cref{thm:dyner_formal}, we show that for unweighted undirected expander, such result carries over to the dynamic setting.
\begin{theorem}[Dynamic effective resistance in expanders]
\label{thm:dyner_formal}
There is a data structure for a fully dynamic unweighted undirected $\gamma$-expander with $O(\gamma^{-4}\eps^{-2}\log n \log^4 \frac{n}{\eps})$ expected time per update that supports $(1+\eps)$-approximate effective resistance and commute time estimation for each pair of vertices in $O(1)$ expected time.
It succeeds with high probability and has preprocessing time and space usage of $O(m+n\gamma^{-3}\eps^{-2}\log n \log^3 \frac{n}{\eps})$.
\end{theorem}

Our dynamic data structure is based on a fully dynamic algorithm for maintaining short random walks.
For the remainder of the proof, we first show that for unweighted undirected graphs (not necessarily expanders), there is an efficient data structure for maintaining uniformly sampled short random walks starting from each vertex.
We then apply this dynamic random walk data structure on expanders to prove \Cref{thm:dyner_formal}.

\subsection{Fully Dynamic Short Random Walks}

We define $(\ell,s)$ random walks as follows: each vertex has $s$ uniform lazy random walks of length $k$ for all $k \in \{0,\ldots,\ell-1\}$.
We show that a uniform sample of such $(\ell,s)$ random walks can be maintained efficiently.
Our main algorithm is presented in Algorithm~\ref{alg:dyn_rw}.
Note that this algorithm also applies for directed graphs.

\begin{algorithm2e}[t]
\caption{Dynamic $(\ell,s)$ random walks}
\label{alg:dyn_rw}
\DontPrintSemicolon
\Pro{\Initialize{$G=(V,E), (\ell,s)$}}{
\For{each vertex $u$ and each length-index pair $(k,i)\in\{0,\ldots,\ell-1\}\times[s]$}{
    Perform a length $k$ random walk $\ww$ starting from $u$\;
    Record the walk as incident to every vertex it visits and every edge it traverses only at non-lazy steps\;
}
}

\Pro{\Insert{$e=(v,w)$}}{
\For{each endpoint $a$ from which $e$ can be traversed ($v$ if directed, both endpoints if undirected)}{
For each instance that a walk is incident with $a$, sample this instance (walk and position) with probability $\frac{1}{d_a+1}$\;
If a walk is sampled multiple times, keep the instance with the earliest position only\;
\For{each remaining instance}{
    Take a step of edge $e$ from this position at $a$\;
    Resample the remaining lazy-walk suffix from the other vertex of $e$\;
    Update vertex and edge incidences\;
}
}
}

\Pro{\Delete{$e=(v,w)$}}{
\For{each walk incident to $e$ (in either direction if undirected)}{
    Regenerate the portion of the walk from the first position where $e$ is traversed, with the first step non-lazy\;
    Update vertex and edge incidences of the walk\;
}
Clear the incidence data structure for $e$\;
}
\end{algorithm2e}

\begin{lemma}[Undirected fully dynamic $(\ell,s)$-random walks]
\label{lemma:dyn_rw}
The data structure in Algorithm~\ref{alg:dyn_rw} maintains a set of uniform $(\ell,s)$-random walks for a fully dynamic undirected graph in $O(s\ell^3)$ expected time and $O(s\ell^2)$ expected recourse per update.
It has preprocessing time and space usage of $O(m + ns\ell^2)$.
\end{lemma}

For a vertex $v\in V$, its \emph{vertex congestion} is defined to be the total number of times a random walk visits $v$.
Similarly, for an edge $e\in E$, its \emph{edge congestion} is the total number of times a random walk traverses $e$.
Note that for undirected graphs, the traversal can happen in both directions.

These congestion measures on vertices and edges are crucial to the control on the update time.
Intuitively, the less congestion, the fewer walks are affected by the insertion or deletion of an edge.
We demonstrate below that for unweighted graphs, these measures are small.
\begin{lemma}[Undirected vertex congestion]
    \label{lemma:und_cong}
For $(\ell,s)$-random walks on an undirected graph with $n$ vertices, in expectation, the congestion for a vertex $v \in V$ is $O(s\ell^2 d_v)$ and for any edge is $O(s \ell^2)$.
\end{lemma}
\begin{proof}
Let $\dd$ be the degree vector. Since $\PP\dd=\dd$ and
$\dd\ge\vone$ coordinatewise, the nonnegativity of $\PP$ implies that, for
every $t\ge0$,
\[
    \PP^t\vone\le\PP^t\dd=\dd.
\]
Therefore, the expected number of visits to a vertex $v$ is
\begin{align*}
    s\sum_{u\in V}\sum_{k=0}^{\ell-1}\sum_{t=0}^{k}
        (\PP^t\ee_u)_v
    &=s\sum_{k=0}^{\ell-1}\sum_{t=0}^{k}(\PP^t\vone)_v \\
    &\le O(s\ell^2d_v).
\end{align*}

At any occurrence of $v$ before the last step, a lazy walk traverses any
fixed incident edge with probability $\frac{1}{2d_v}$.
Hence, for an undirected
edge $e=\{v,w\}$, summing over its two traversal directions gives expected
edge congestion at most
\[
    \frac{1}{2d_v}O(s\ell^2d_v)
    +\frac{1}{2d_w}O(s\ell^2d_w)
    =O(s\ell^2).
\]
\end{proof}

We are ready to prove \Cref{lemma:dyn_rw}.
\begin{proof}[Proof of \Cref{lemma:dyn_rw}]
For the data structure that maintains incident walks for vertices and edges, we use dynamic arrays and hash tables supporting entry insertion and deletion in $O(1)$ expected time.
It also supports, by a fast sampling algorithm (e.g.~\cite{Knuth97,Devroye06,BringmannP12}), sampling each entry of the incident set independently with a fixed probability $p\in(0,1]$ in time proportional to the sampled subset, which is $O(p|I|+1)$ in expectation for the incident set $I$ of size $|I|$.

For initialization, the total number of sampled transitions is
$O(ns\ell^2)$.  Sampling a transition and recording its incidences take expected $O(1)$ time, so storing all walks takes $O(ns\ell^2)$ time and space. 
Reading and storing the initial graph takes an additional $O(m)$ time and space.

Assume inductively that all maintained walks are mutually independent and distributed as fresh lazy random walks in the current graph. 
The next update is independent of the data structure's random choices.

\paragraph{Edge insertions.}
Let $e=\{v,w\}$ be inserted, and fix an endpoint $a\in\{v,w\}$. 
Let $d_a$ be its degree before the insertion. 
Conditional on a transition from $a$ being non-lazy, each old incident edge has probability $1/d_a$. 
Our algorithm keeps the old
edge with probability $d_a/(d_a+1)$ and otherwise uses $e$.  Every old edge and the new edge then have probability $1/(d_a+1)$ as required, while leaving the laziness coin unchanged.

We first sample all incidences to be repaired from the old walks before resampling any suffix.
If a walk contains several selected incidences, its earliest one is the first position where the old and new walks disagree and it suffices to only resample after taking a step with edge $e$.
By the vertex congestion bound in \Cref{lemma:und_cong}, the expected number of resamples is then at most
\[
    \sum_{a\in\{v,w\}}\frac{1}{d_a+1}\,O(s\ell^2d_a)
    =O(s\ell^2).
\]
Resampling a suffix and updating its incidences take
$O(\ell)$ time, giving $O(s\ell^2)$ expected recourse and
$O(s\ell^3)$ expected update time.

\paragraph{Edge deletions.}
Let $e=\{v,w\}$ be deleted.  A walk is modified if and only if it traverses $e$, and only its earliest traversal needs to be retained. 
Suppose this traversal leaves a vertex $a$ of degree $d_a$ before deletion. The first resampling step is non-lazy and chooses a remaining edge uniformly since the previous laziness coin is still valid.
For every remaining edge, the resulting conditional probability is
\[
    \frac{1}{d_a}
    +\frac{1}{d_a}\cdot\frac{1}{d_a-1}
    =\frac{1}{d_a-1}.
\]
Thus, keeping the prefix and independently resampling after the replacement again gives fresh, mutually independent walks in the updated graph.

By the edge congestion bound in \Cref{lemma:und_cong} above, the expected number of resamples is $O(s\ell^2)$.  
The $O(\ell)$ suffix repair cost gives $O(s\ell^2)$ expected recourse and $O(s\ell^3)$ expected update time. 
\end{proof}

\subsection{Proof of \Cref{thm:dyner_formal}}
The effective resistance between a pair of vertices can be computed exactly through the difference hitting vector $\hh_{u,v}$.
\begin{proposition}\label{prop:er_id}
\[
    \ER_G(u,v)=\frac{1}{2}(\ee_u-\ee_v)^\top\DD^{-1}\hh_{u,v}
\]
\end{proposition}
\begin{proof}
Set $\bb=\ee_u-\ee_v$.  Since $\vone^\top\bb=0$, $\hh_{u,v}$ converges and $(\II-\PP)\hh_{u,v}=\bb$.
Moreover,
\[
    \II-\PP
    =\frac{1}{2}(\II-\AA\DD^{-1})
    =\frac{1}{2}\LL\DD^{-1}.
\]
Thus $\frac{1}{2}\DD^{-1}\hh_{u,v}$ is a solution to $\LL\xx=\bb$.
Every such solution differs from $\LL^\dagger\bb$ by a multiple of $\vone$, which is orthogonal to $\bb$. Therefore,
\[
    \frac{1}{2}\bb^\top\DD^{-1}\hh_{u,v}
    =\bb^\top\LL^\dagger\bb
    =\ER_G(u,v).
\]
\end{proof}

For a sampled set of $(\ell,s)$ random walks, we define $\widetilde{\hh}_v^{(\ell)}$ to be the estimator of $\hh_v^{(\ell)}$
and set 
\[
    \ssig_v = \widetilde{\hh}^{(\ell)}_v - \ell \cdot \ppi,
\]
where $\ppi = \frac{1}{\sum_v d_v} \dd$ is the stationary vector.

Given as query a pair of vertices $(v,w)$ the algorithm of \cite{LiS23} simply returns
\begin{equation}
    \frac{1}{2}\left(
    \frac{1}{d_v} \ssig_v(v) - \frac{1}{d_w} \ssig_v(w)
    - \frac{1}{d_v} \ssig_w(v) + \frac{1}{d_w} \ssig_w(w)
    \right).
\end{equation}
Since the stationary terms cancel in this expression,
our implementation stores only the sparse hitting vectors and uses the stationary term only in the approximation analysis.
The overall static construction is stated in Algorithm~\ref{alg:er_static}.

\begin{algorithm2e}[t]
\caption{Static all-pairs effective resistances}
\label{alg:er_static}
\Pro{\Initialize{$G=(V,E), (\ell,s)$}}{
    \For{$v \in V$}{
        $\widetilde{\hh}^{(\ell)}_v \gets \vzero$\;
        \For{$k \in \{0,\ldots,\ell-1\}, i \in [s]$}{
            perform a length-$k$ random walk starting from $v$, ending at vertex $w$\;
            $\widetilde{\hh}^{(\ell)}_v \gets \widetilde{\hh}^{(\ell)}_v + \frac{1}{s} \ee_w$\;
        }
    }
}
\end{algorithm2e}

\begin{proof}[Proof of \Cref{thm:dyner_formal}]
By \cite{LiS23}, it suffices to choose
$\ell=O(\gamma^{-1}\log\frac{n}{\eps})$
and $s=O(\eps^{-2}\gamma^{-1}\log n \log\frac{n}{\eps})$.
By \Cref{lemma:dyn_rw}, the random walks can be initialized in time and space \[
    O(m+ns\ell^2)
    =O\left(m+n\gamma^{-3}\eps^{-2}\log n
    \log^3\frac{n}{\eps}\right).
\]

For each vertex $v$, store the sparse hitting vector
$\widetilde{\hh}^{(\ell)}_v$. 
Whenever a maintained walk changes, subtract $\frac{1}{s}$ from its old terminating location and add $\frac{1}{s}$ to its new hitting location.
Delete entries when they become zero.
By the recourse bound in \Cref{lemma:dyn_rw}, this costs $O(s\ell^2)$ expected time per update and is
dominated by the random-walk update time
\[
    O(s\ell^3)
    =O\left(\gamma^{-4}\eps^{-2}\log n
    \log^4\frac{n}{\eps}\right).
\]

The maintained walks remain uniform on the current graph.  Thus, the static guarantee of \cite{LiS23} gives the claimed approximation at any fixed time with very high probability.

Each hitting vector is stored in a hash table. 
A query reads four coordinates in these vectors and two current degrees and therefore takes expected $O(1)$ time. 
For commute time queries, return $2m$ times the effective resistance estimate, where $m=|E|$ is maintained in $O(1)$ time per update.
The sampled hitting vectors use $O(ns\ell)$ space, which is dominated by the $O(ns\ell^2)$ space for the walks and their incidences.
\end{proof}

\section{Conditional Hardness for \texorpdfstring{$\kappa$}{kappa}-Mixing Graphs}
\label{sec:hardness}

In this section, we state and formally prove our hardness result for approximate dynamic stationary distribution data structure under the assumption of the OMv conjecture.

\begin{theorem}[Hardness for fast mixing directed graphs]
\label{thm:mixing-hardness}
Assume the OMv conjecture. For every constant $\delta>0$, there is no randomized fully-dynamic data structure with polynomial preprocessing time that maintains an $n$-vertex unweighted directed $O(\log n)$-mixing multigraph of maximum out-degree four, answers $(1+0.5)$ entry-wise approximate stationary queries, and has both amortized update time and
query time $O(n^{1/2-\delta})$ with probability at least $\frac{2}{3}$. 
\end{theorem}

To prove the theorem above, we give a simple reduction from boolean matrix vector products to a fast mixing directed graph of roughly the same size.
An entry of the product can be determined with good probability of success by checking if the stationary distribution of the corresponding vertex crosses a certain threshold.

\subsection{From Matrix Vector Product to Fast Mixing Graphs}

\begin{figure}[t]
\centering
\begin{tikzpicture}[scale=1.3,>=Stealth,
  vtx/.style={circle, draw, inner sep=1pt, minimum size=14pt, font=\small},
  lab/.style={font=\scriptsize},
  vec-edge/.style={->, very thick, blue!70!black},
  mat-edge/.style={->, very thick, red!70!black},
  rst/.style={->, dashed, gray!80!black}
]
\node[vtx] (s) at (0,0) {$s$};
\draw[->] (0.33,0.06) -- (1.5,1.35);
\draw[->] (0.33,-0.06) -- (1.5,-1.35);
\draw (1.5,1.35) -- (1.5,-1.35);
\node[vtx] (c1) at (1.9,1) {$c_1$};
\node[vtx] (cj) at (1.9,0) {$c_j$};
\node[vtx] (cq) at (1.9,-1) {$c_q$};
\node at (1.9,0.52) {$\vdots$};
\node at (1.9,-0.48) {$\vdots$};
\draw[->] (1.5,1) -- (c1);
\draw[->] (1.5,0) -- (cj);
\draw[->] (1.5,-1) -- (cq);
\node[vtx] (x1) at (3.3,1) {$x_1$};
\node[vtx] (xj) at (3.3,0) {$x_j$};
\node[vtx] (xq) at (3.3,-1) {$x_q$};
\node at (3.3,0.52) {$\vdots$};
\node at (3.3,-0.48) {$\vdots$};
\draw[vec-edge] (c1) -- (x1);
\draw[vec-edge] (cj) -- (xj);
\draw[vec-edge] (cq) -- (xq);
\node[lab, blue!70!black] at (2.6,1.55) {vector edge iff $\vv_j = 1$};
\draw[->] (3.63,1.05) -- (4.9,1.45);
\draw[->] (3.63,0.95) -- (4.9,0.55);
\draw (4.9,1.45) -- (4.9,0.55);
\draw[->] (3.63,0.05) -- (4.9,0.45);
\draw[->] (3.63,-0.05) -- (4.9,-0.45);
\draw (4.9,0.45) -- (4.9,-0.45);
\draw[->] (3.63,-0.95) -- (4.9,-0.55);
\draw[->] (3.63,-1.05) -- (4.9,-1.45);
\draw (4.9,-0.55) -- (4.9,-1.45);
\node[vtx] (m11) at (5.4,1) {$m_{1,1}$};
\node[vtx] (mij) at (5.4,0) {$m_{i,j}$};
\node[vtx] (mqq) at (5.4,-1) {$m_{q,q}$};
\node at (5.4,0.52) {$\vdots$};
\node at (5.4,-0.48) {$\vdots$};
\draw[->] (4.9,1) -- (m11);
\draw[->] (4.9,0) -- (mij);
\draw[->] (4.9,-1) -- (mqq);
\node[vtx] (y11) at (6.9,1) {$x_{1,1}$};
\node[vtx] (yij) at (6.9,0) {$x_{i,j}$};
\node[vtx] (yqq) at (6.9,-1) {$x_{q,q}$};
\node at (6.9,0.52) {$\vdots$};
\node at (6.9,-0.48) {$\vdots$};
\draw[mat-edge] (m11) -- (y11);
\draw[mat-edge] (mij) -- (yij);
\draw[mat-edge] (mqq) -- (yqq);
\node[lab, red!70!black] at (6.15,1.55) {matrix edge iff $\MM_{i,j} = 1$};
\draw (7.3,1.45) -- (7.3,0.55);
\draw[->] (7.3,1.45) -- (7.72,1.07);
\draw[->] (7.3,0.55) -- (7.72,0.93);
\draw (7.3,0.45) -- (7.3,-0.45);
\draw[->] (7.3,0.45) -- (7.72,0.07);
\draw[->] (7.3,-0.45) -- (7.72,-0.07);
\draw (7.3,-0.55) -- (7.3,-1.45);
\draw[->] (7.3,-0.55) -- (7.72,-0.93);
\draw[->] (7.3,-1.45) -- (7.72,-1.07);
\node[vtx] (r1) at (8.25,1) {$r_1$};
\node[vtx] (ri) at (8.25,0) {$r_i$};
\node[vtx] (rq) at (8.25,-1) {$r_q$};
\node at (8.25,0.52) {$\vdots$};
\node at (8.25,-0.48) {$\vdots$};
\draw[->] (7.78,1) -- (r1);
\draw[->] (7.78,0) -- (ri);
\draw[->] (7.78,-1) -- (rq);
\node[vtx] (t) at (10.15,0) {$t$};
\draw (8.7,1.35) -- (8.7,-1.35);
\draw[->] (8.7,1.32) -- (9.88,0.08);
\draw[->] (8.7,-1.32) -- (9.88,-0.08);
\draw[->] (r1) -- (8.7,1);
\draw[->] (ri) -- (8.7,0);
\draw[->] (rq) -- (8.7,-1);
\draw[rst] (yij.north) .. controls (5.8,3.3) and (1.0,3.2) .. (s.north);
\node[lab, gray!80!black] at (3.3,3.45) {reset edges to $s$: $\max\{1,d_{fo}\}$ parallel edges at each $v \neq s$ (only one shown)};
\node[vtx] (sp) at (4.3,-3.45) {$s'$};
\fill (-0.55,-1.1) circle (1.2pt);
\fill (0.1,-2.15) circle (1.2pt);
\fill (1.6,-3.3) circle (1.2pt);
\draw[->] (s.south) -- (-0.55,-1.1);
\draw[->] (-0.55,-1.1) -- (0.1,-2.15);
\node at (0.8,-2.88) {$\cdots$};
\draw[->] (1.6,-3.3) -- (sp);
\node[lab] at (0.55,-3.68) {path of length $2\log q + 4$};
\draw[->] (sp.168) -- (1.9,-2.05);
\draw[->] (sp.12) -- (8.25,-2.05);
\draw (1.9,-2.05) -- (8.25,-2.05);
\draw[->] (1.9,-2.02) -- (cq.south);
\draw[->] (3.3,-2.02) -- (xq.south);
\draw[->] (5.4,-2.02) -- (mqq.south);
\draw[->] (6.9,-2.02) -- (yqq.south);
\draw[->] (8.25,-2.02) -- (rq.south);
\node[lab] at (4.95,-2.62) {out-tree with leaves $V'_{\mathrm{main}}$};
\end{tikzpicture}
\caption{Schematic of the product graph $G_{\MM,\vv}$. Triangles in the main component denote complete binary trees of the indicated depth. Blue edges are the vector edges $(c_j, x_j)$, present iff $\vv_j = 1$; they are the only edges updated by the dynamic algorithm. Red edges are the matrix edges $(m_{i,j}, x_{i,j})$, present iff $\MM_{i,j} = 1$. A path $s \leadsto c_j \to x_j \leadsto m_{i,j} \to x_{i,j} \leadsto r_i \leadsto t$ through a vector edge exists iff $[\MM\vv]_i = 1$. The auxiliary gadget at the bottom makes the graph strongly connected: $s$ reaches $s'$ through a long path, and every product vertex has a copy at a leaf of the tree rooted at $s'$, with an edge from each copy to its original vertex.
Each vertex other than $s$ has reset edges back to $s$, signified by the grey dashed arc. These edges give the rapid mixing property.}
\label{fig:prod-graph}
\end{figure}

Let $q \ge 4$ be a power of two. For $\MM \in \{0,1\}^{q \times q}$ and $\vv \in \{0,1\}^q$, we define the directed multigraph $G_{\MM,\vv}$ as follows (see \cref{fig:prod-graph} for a schematic).

For the main component of the graph, we define for each $i \in [q]$ a path for the product $[\MM\vv]_i$.
Its vertex set consists of a source $s$, a sink $t$, and a collection $V_{\mathrm{main}}$ of \emph{product vertices}:
\begin{itemize}
    \item \emph{column vertices} $c_j, x_j$ for $j \in [q]$, and \emph{row vertices} $r_i$ for $i \in [q]$;
    \item \emph{matrix vertices} $m_{i,j}, x_{i,j}$ for $(i,j) \in [q]\times[q]$;
    \item the internal vertices of four families of complete binary trees, all of depth exactly $\log q$: an out-directed tree rooted at $s$ with leaves $c_1,\ldots,c_q$; an in-directed tree rooted at $t$ with leaves $r_1,\ldots,r_q$; for each $j \in [q]$, an out-directed tree rooted at $x_j$ with leaves $m_{1,j},\ldots,m_{q,j}$; and for each $i \in [q]$, an in-directed tree rooted at $r_i$ with leaves $x_{i,1},\ldots,x_{i,q}$.
\end{itemize}
The graph has the following directed forward edge connections:
\begin{itemize}
    \item the \emph{vector edge} $(c_j,x_j)$ for each $j$ with $\vv_j = 1$;
    \item the \emph{matrix edge} $(m_{i,j},x_{i,j})$ for each $i,j$ with $\MM_{i,j} = 1$.
\end{itemize}
The vector edges within $\{(c_j,x_j) | j \in [q]\}$ are the only edges being updated after receiving a new query vector $\vv$.

To ensure all vertices are reachable from $s$, the following additional \emph{forward connection} vertices and edges are included:
\begin{itemize}
    \item a second source $s'$, connected from $s$ by a path of length $\ceil{2\log q} + 4$;
    \item an out-directed binary tree with leaves at depth $\ceil{\log |V_{\mathrm{main}}|}$ or $\ceil{\log |V_{\mathrm{main}}|}-1$ from $s'$ to a leaf set $V'_{\mathrm{main}}$ consisting of exactly one new copy of each vertex in $V_{\mathrm{main}}$.
    \item each copied leaf has a directed edge to its corresponding vertex in $V_{\mathrm{main}}$.
\end{itemize}
Finally, for strong-connectivity and fast-mixing, we include the following backward \emph{reset} edges directed to the source $s$:
\begin{itemize}
    \item for each vertex $v \ne s,t$, if its forward out-degree is $d_{fo}$, include $\max\{1,d_{fo}\}$ reset edges $(v,s)$;
    \item a single edge $(t,s)$.
\end{itemize}

For any $\MM,\vv$, the graph $G_{\MM,\vv}$ has size roughly equal to the size of the matrix vector product instance and has maximum out-degree of at most 4.
\begin{lemma} \label{lem:omv_size}
For any $\MM \in \{0,1\}^{q \times q}$ and $\vv \in \{0,1\}^q$, the graph $G_{\MM,\vv}$ has $N=\Theta(q^2)$ vertices with out-degree at most $4$ for all vertices.
\end{lemma}
\begin{proof}
The source and sink, and the vertices in $V_{\mathrm{main}}$, contain $4q^2+q-2$ vertices. 
A full binary tree rooted at $s'$ with one new leaf for each vertex of $V_{\mathrm{main}}$ has $2|V_{\mathrm{main}}|-1$ vertices. 
The path from $s$ to $s'$ contains $O(\log q)$ vertices.
The total is therefore $\Theta(q^2)$.

Every vertex has at most two non-reset outgoing edges.
The only exception is that $s$ has three. 
Reset edges at most double the out-degree of vertices with positive forward out-degree, and vertices with forward out-degree zero have one reset edge. The maximum out-degree is then $4$.
\end{proof}

The reset edges have two specific consequences. 
First, the reset edges force the graph to have a rather skewed stationary distribution.
For the source vertex $s$, the stationary mass is at least $\frac{1}{3}$; we prove this formally in \Cref{lem:hard_stamass}.
Secondly, since a reset happens with probability at least $\frac{1}{4}$, the random walk mixes very rapidly.
Intuitively, for two starting points, a single random walk step would funnel them to the same location with constant probability.
Therefore, the remaining difference in the stationary mass between the two starting distributions decreases by a constant factor after each step.
We prove this observation formally in \Cref{lem:omv_mixing}.
\begin{lemma} \label{lem:omv_mixing}
The graph $G_{\MM,\vv}$ is $O(\log q)$-mixing.
\end{lemma}
\begin{proof}
The forward edges give connection from $s$ to every other vertex in the graph, while the reset edges give direct connection from those vertices back to $s$.
Therefore, the graph is strongly connected.

We show lazy random walks on $G_{\MM,\vv}$ mix quickly by showing that each step of the walk makes constant progress towards the stationary distribution using a standard random walk coupling argument.
Since for each vertex other than $s$, at least half of the edges are reset edges, a lazy walk step visits $s$ with probability at least $\frac{1}{4}$.
Also, a walk currently at $s$ stays there with probability $1/2$ by the lazy step.
Hence the one-step random walk distributions from any two vertices have common mass at least $\frac{1}{4}$ at $s$.

We now couple a walk starting from an arbitrary vertex with a walk starting from each vertex with the stationary probability. 
At each step, conditional on not having coupled earlier, the probability that they have not coupled is at most $\frac{3}{4}$.
For $k$ steps, the coupling gives $\frac12\lVert \PP^k\ee_x-\ppi\rVert_1\le (3/4)^k$.
Taking $k\ge 100\log n$ makes the error at most $n^{-10}$ for every entry. 
Since $n=\Theta(q^2)$, this gives $\kappa=O(\log q)$-mixing, proving the lemma.
\end{proof}

Within the main graph, a product entry $[\MM\vv]_i$ is true if and only if there exists a path from the source $s$ to the corresponding vertex $r_i$ of entry $i \in [q]$.
Due to the long path between $s$ and $s'$, where each step resets with probability at least $\frac{1}{4}$, a random walk is much more likely to pick the main path from $s$ to $r_i$ as opposed to going through the path using the binary tree rooted at $s'$. 
Consequently, a higher stationary mass would be placed at $r_i$ compared to the case where no main path from $s$ to $r_i$ exists.
A $(1+0.5)$-approximation to the stationary probability would then suffice for distinguishing the value of $[\MM\vv]_i$.
This is formalized in the following \Cref{lem:hard_stamass}.

\begin{lemma} \label{lem:hard_stamass}
Let $\ppi$ be the stationary distribution of the graph $G_{\MM,\vv}$.
It satisfies that for any $i \in [q]$
\[
\ppi_{r_i}
\begin{cases}
   \le \frac{1}{75q^5} & \text{if } [\MM \vv]_i =0  \\
   \ge \frac{1}{9q^5} & \text{if } [\MM \vv]_i =1  \\
\end{cases}
\]
\end{lemma}
\begin{proof}

Since the stationary distribution is not affected by the laziness of the random walks, we assume w.l.o.g. within this proof that the random walks are non-lazy.

Suppose $[\MM\vv]_i = 1$ and that there exists an index $j$ with $\MM_{ij} = \vv_j = 1$ witnessing this.
Then, there is a main path from $s$ to $r_i$ through $c_j,x_j$, and $m_{i,j},x_{i,j}$ and the binary tree paths connecting them.
The total length of the binary tree edges is $3\log q$.
After the first step from $s$, the path has $2\log q-1$ out-tree edges with sampling probability $1/4$ each, and $\log q$ in-tree edges and two extra edges corresponding to matrix and vector bits with sampling probability $1/2$ each.
Starting from $s$, the chance of following this main
path before a reset is then at least
\[
\frac13 \cdot \left(\frac14\right)^{2\log q-1}\cdot\left(\frac12\right)^{\log q+2} =\frac{1}{3q^5},
\]
where $\frac{1}{3}$ is the probability of taking the first edge of this path from $s$.
As random walks return to $s$ with probability at least $\frac{1}{2}$, we can crudely bound $\ppi_s \ge \frac{1}{3}$, which gives $\ppi_{r_i} \ge \frac{1}{9q^5}$.

Now suppose $[\MM\vv]_i=0$. 
A path from $s$ can reach $r_i$ only through the added forward tree.
The chance of reaching $s'$ from $s$ is $\frac{1}{24q^2}$.
Given that the walk is at $s'$, the probability of reaching the leaf assigned to any fixed vertex of $V_{\mathrm{main}}$ is at most $\frac{1}{16q^4}$.
The in-tree for row $i$ contributes $\log q+1$, one from each level.
The matrix vertices for row $i$ contribute at most $1/2$, and their column trees contribute at most $\frac12\sum_{d=1}^{\log q}4^{-d}\le \frac 1 6$.
All other vertices contribute zero in the zero case.
As $s$ has no self-loop, the expected return time to $s$ is at least $2$. Thus, $\ppi_s \le \frac{1}{2}$.
We then have
\[
 \ppi_{r_i}
 \leq \frac12\cdot\frac{1}{24q^2}\cdot\frac{1}{16q^4}
       \left(\log q+\frac53\right)
 \leq \frac{1}{75q^5},
\]
where the last inequality holds for $q\geq4$.
\end{proof}

\subsection{Proof of \Cref{thm:mixing-hardness}}

Similar to OMv reductions standard in the literature~\cite{HenzingerKNS15,HenzingerPS22}, we simulate an OMv instance by partitioning it into evenly sized blocks and applying OR over all products on the blocks.
This construction allows us to bypass any large polynomial preprocessing time to argue on the update and query time alone.

\begin{proof}[Proof of \Cref{thm:mixing-hardness}]
We assume w.l.o.g. that the dimension $n$ of the OMv instance is a power-of-two, padding the matrix and vectors with zeros if necessary.
We take $4\le q\le n$ to be a power of two (rounding down if necessary) and set $r=\frac n q$.
We partition the matrix $\MM$ into $r^2$ blocks of size $q \times q$.
To ensure the $\frac{2}{3}$ probability of success, for each block, we keep $R=\Theta(\log n)$ copies from which we take the $\text{median}$.
We assume also there is a constant $c \ge 2$ such that a dynamic $(1+0.5)$-approximate stationary data structure preprocesses in time $O(q^c)$.

In every copy for a block, we use at most $q$
updates for each sub-vector of $\vv$, and query its $q$ row vertices.
We take the median estimate for each row vertex and test it against the threshold $\frac{1}{24q^5}$. 
By \Cref{lem:hard_stamass} and $(1+0.5)$-approximate, this recovers $\MM^{h,\ell}\vv^\ell$, since it satisfies 
\[
    \frac{2}{75q^5}<\frac{1}{24q^5}<\frac{1}{2 \cdot 9q^5}.
\]
To get the final product, we take an OR over all blocks for each row $i$.

There are $Rr^2$ copies of the dynamic approximate stationary data structures in total. 
With the parameter choice above, their total preprocessing time is
    $O(Rr^2q^c)
    =\widetilde O(n^2q^{c-2}).
    $
For $q = \Theta(n^{1/2c})$, this is $\widetilde O(n^{2.5-1/c})$, which is strictly less than $O(n^{2.5})$.
Constructing the graphs themselves takes only
$O(Rr^2q^2)=\widetilde O(n^2)$ additional time.

For each of the $n$ online vectors, every copy receives at most $q$ updates and $q$ queries. 
Hence the total numbers of updates and queries are both $O(Rnr^2q)=\widetilde O(n^3/q).$
Computing all medians takes $O(Rnr^2q) = \tO(n^3/q)$ time, and the entry-wise ORs take $O(n^3/q)$ total time, which are both subsumed by the total number of updates and queries.

By \Cref{lem:omv_size,lem:omv_mixing}, each maintained graph has $N=\Theta(q^2)$ vertices, maximum out-degree four, and mixing time $O(\log N)$ always.

By our choice of $R=\Theta(\log n)$, a Chernoff bound and a union bound over at most $n^3$ medians give an overall success probability of at least $\frac{2}{3}$ as required.

We assume for contradiction that the update and query bounds are at most
\[
    O(N^{1/2-\delta})=O(q^{1-2\delta}).
\]
It suffices to consider any $0<\delta\le1/4$.
Consequently, the total online time is
\[
    \widetilde O\!\left(\frac{n^3}{q}\,q^{1-2\delta}\right)
    =\widetilde O(n^3q^{-2\delta})
    =\widetilde O\bigl(n^{3-\delta/c}\bigr).
\]
Taking, for simplicity, $\eta = \frac{\delta}{2c} > 0$, the overall processing and online computation take $O(n^{3-\eta})$ time. 
This contradicts the OMv conjecture (\Cref{conj:omv}) and proves the theorem.
\end{proof}

\section{Fully Dynamic Algorithms for \texorpdfstring{$\kappa$}{kappa}-Uniform Graphs}
\label{sec:dir}

Under the additional requirement that an unweighted directed $\kappa$-mixing graph has almost uniform distribution, the hard instance for designing efficient dynamic data structures can be avoided.
We show that for such $\kappa$-uniform graphs, one can design efficient algorithms for both computing entry-wise approximate stationary distribution (see \Cref{thm:dynsta_formal}) and all-pairs commute time (see \Cref{thm:dyncomm-formal}), again, using dynamic short random walks.

\begin{theorem}[Dynamic stationary approximation for $\kappa$-uniform graphs]
\label{thm:dynsta_formal}
There is a data structure for a fully dynamic unweighted directed $\kappa$-uniform graph $G$ with expected time $O(\kappa^6 \eps^{-2}\frac{ \log^3(\kappa n \eps^{-1})}{\log^2 n})$ per update that supports a $(1+\eps)$-approximate stationary entry query in $O(1)$ expected time.
It succeeds with high probability and has preprocessing time and space complexity of $O(m+n \kappa^3\eps^{-2}\frac{ \log^2(\kappa n \eps^{-1})}{\log n})$ in expectation.
\end{theorem}

\begin{theorem}[Dynamic commute time for $\kappa$-uniform graphs]
\label{thm:dyncomm-formal}
For $0<\eps<1/2$, there is a data structure for a fully dynamic unweighted directed $\kappa$-uniform graph $G$ with expected time
$O(\kappa^6 \eps^{-2}\frac{ \log^4(\kappa n \eps^{-1})}{\log^3 n})$ per update
that supports $(1+\eps)$-approximate commute time queries in $O(1)$ expected time.
It succeeds with high probability over an oblivious polynomial length sequence and has preprocessing time and space complexity of $O(m+n\kappa^3\eps^{-2}\frac{\log^3(\kappa n\eps^{-1})}{\log^2 n})$.
\end{theorem}

\subsection{Dynamic Random Walks on $\kappa$-Uniform Graphs}
On $\kappa$-uniform directed graphs, we show that Algorithm~\ref{alg:dyn_rw} suffices for maintaining $(\ell,s)$-random walks, analogous to \Cref{lemma:dyn_rw}.
\begin{lemma}[Directed $\kappa$-uniform fully dynamic $(\ell,s)$-random walks]
\label{lemma:dir_dyn_rw}
The data structure in Algorithm~\ref{alg:dyn_rw} maintains a set of uniform $(\ell,s)$-random walks for a fully dynamic unweighted directed $\kappa$-uniform graph in $O(\kappa^2s\ell^3)$ expected time and $O(\kappa^2s\ell^2)$ expected recourse per update.
It has preprocessing time and space usage of $O(m + ns\ell^2)$.
\end{lemma}

The crux of the proof of \Cref{lemma:dir_dyn_rw} is then the following congestion guarantee, specialized for $\kappa$-uniform graphs.
\begin{lemma}[Congestion bound for $\kappa$-uniform graphs]
\label{lemma:dir_cong}
For $(\ell,s)$-random walks on a $\kappa$-uniform graph with $n$ vertices, in expectation the congestion is $O(s\ell^2 \kappa^2)$ for any vertex and edge.
\end{lemma}
\begin{proof}
For every vertex $x$ and $k\geq 0$, the $\kappa$-uniformity of the stationary distribution gives
\[
    \sum_{u\in V}(\PP^k)_{xu}
    \leq \sum_{u\in V}(\PP^k)_{xu}(\kappa n \cdot \ppi_u)
    =\kappa n\ppi_x
    \leq \kappa^2.
\]
Then, for a fixed walk length $k$, the expected number of visits to $x$ at each step, summed over all starting vertices and all $s$ copies, is then at
most $s\kappa^2$. 
There are $O(\ell^2)$ total walk steps over all lengths, which proves the vertex bound.
A walk can cross an edge only after visiting
its tail, so the same bound holds for every edge.
\end{proof}

\begin{proof}[Proof sketch of \Cref{lemma:dir_dyn_rw}]
The correctness of the algorithm follows almost identically to that of \Cref{lemma:dyn_rw}.

It remains only to adjust the running time and space bounds. 
These are even more straightforward here: taking the bound from \Cref{lemma:dir_cong} suffices, giving
$O(\kappa^2s\ell^2)$ expected recourse and
$O(\kappa^2s\ell^3)$ expected update time.
Initialization samples $O(ns\ell^2)$ random walk steps in total, which gives the claimed preprocessing time and space.
\end{proof}

\subsection{Dynamic Stationary Approximation}
In this section, we prove \Cref{thm:dynsta_formal}.
We show first that with a small $\tO(\kappa)$ length, the distribution obtained from a uniform starting vertex converges entry-wise very closely to the stationary distribution.

\begin{lemma}[Mixing length] \label{lem:kappa_mix_len} 
For any $\kappa$-uniform graph $G$ and any $0 < \eps < \frac{1}{2}$, for any length 
$\ell \ge C \kappa \frac{\log(\kappa n \eps^{-1})}{\log n}$ for a sufficiently large constant $C$, for all $v \in V$,
\[
\Abs{\frac{1}{n} (\PP^\ell \vone)_v - \ppi_v} \le \eps \ppi_v
\]
\end{lemma}
\begin{proof}
Up to a scaling by a constant, we may assume $\kappa$ divides $\ell$.
The $\kappa$-mixing assumption implies, for every $\xx$ with $\vone^\top\xx=0$, that
\[
\norm{\PP^\kappa\xx}_1
\le \sum_u|\xx_u|\norm{\PP^\kappa\ee_u-\ppi}_1
\le n^{-9}\norm{\xx}_1.
\]
Then, induction and averaging over all vertices give
\[
\Abs{\frac1n(\PP^\ell\vone)_v-\ppi_v}
\le 2n^{-9\ell/\kappa}\le\frac{\eps}{\kappa n}\le\eps\ppi_v,
\]
where the second inequality follows from the lower bound on $\ell$ in the lemma statement for sufficiently large $C$, and the last follows from $\kappa$-uniformity.
\end{proof}

Furthermore, we show that for such short random walks, one can obtain a high-accuracy estimation for the hitting vectors using only a small set of length $\ell$ walks starting from each vertex.
\begin{lemma}[Hitting vector concentration]
\label{lem:stationary-endpoints}
For a $\kappa$-uniform graph, assume that $\ell$ satisfies the bound in \Cref{lem:kappa_mix_len} and sample for each vertex $s$ walks of length $\ell$. 
If $s\geq C'\kappa\eps^{-2}\log n$ for a sufficiently large constant $C'$, then, with high probability, for every
vertex $v$,
\[
\left|\frac1n\sum_u\widetilde{\hh}_u^{(=\ell)}(v)
      -\frac1n\sum_u\hh_u^{(=\ell)}(v)\right|
\leq\eps\ppi_v.
\]
\end{lemma}
\begin{proof}
For vertex $v$, the expected number of random walks of length $\ell$ that terminate at $v$ is exactly $s \sum_u(\PP^\ell)_{vu}$, which, by \Cref{lem:kappa_mix_len} and $\kappa$-uniformity, is at least $\frac{s}{2\kappa}$.
After scaling by a factor of $\frac{1}{n}$ and observing again that $\ppi_v \ge \frac{1}{\kappa n}$ by $\kappa$-uniformity,
a Chernoff bound then gives
\[
 \Pr\left[
 \left|\frac1n\sum_u\widetilde{\hh}_u^{(=\ell)}(v)
       -\frac1n\sum_u\hh_u^{(=\ell)}(v)\right|
       >\eps\ppi_v\right]
 \leq 2\exp\!\left(-\Omega(\eps^2s/\kappa)\right).
\]
For a sufficiently large constant $C'$ in the bound on $s$, a union bound ensures that the inequality holds over all vertices with high probability.
\end{proof}

We now prove the main theorem. 
\begin{proof}[Proof of \Cref{thm:dynsta_formal}]
Our dynamic data structure uses the dynamic $(\ell,s)$-random walks algorithm (Algorithm~\ref{alg:dyn_rw}) with 
\footnote{We use the black-box dynamic data structure for ease of exposition.}
\[
 \ell=\Theta\!\left(
   \kappa\frac{\log(\kappa n\eps^{-1})}{\log n}\right),
 \qquad
 s=\Theta(\kappa\eps^{-2}\log n).
\]
with constants large enough such that \Cref{lem:kappa_mix_len,lem:stationary-endpoints} each hold at length $\ell-1$ with error $\frac{\eps}{2}$.
Then with high probability, the sampled hitting vector for length $\ell-1$ walks gives $(1\pm\eps)$ approximation for the stationary distribution.

\Cref{lemma:dir_dyn_rw} then gives the desired running time and space bounds.
We maintain the normalized sum of hitting vectors over all starting vertices in a dense vector.
Updating this vector takes constant time per resampled walk.
Since a query reads one counter, the query time is $O(1)$.
\end{proof}

\subsection{Dynamic Commute Time Estimations}
In this section, we prove the dynamic commute time result in \Cref{thm:dyncomm-formal}.

To better quantify the commute time and the necessary accuracy, we give a lower bound in \Cref{lem:comm_lb} using the following Kac's return time identity for random walks.
\begin{lemma}[Kac's return time~\protect\cite{LevinP17},
Lemma~21.12] \label{lem:kac-return-time}
For a finite irreducible random walk with stationary distribution $\ppi$ and a vertex $v$, define its first positive return time by $T_v^+=\min\{t\ge1: \text{visits } v \text{ at time } t\}$.
Then, for every $v$,
\begin{equation}
    \E[T_v^+ | \text{starts at } v]=\frac{1}{\ppi_v}.
    \label{eq:kac-return-time}
\end{equation}
\end{lemma}

\begin{lemma}[Commute time lower bound]
\label{lem:comm_lb}
For any $\kappa$-uniform graph $G$ with stationary $\ppi$, for any distinct vertices $u,v$,
\[
C_{u,v} \ge \frac{1}{\min\{\ppi_u,\ppi_v\}} \ge \frac{n}{\kappa}
\]
\end{lemma}
\begin{proof}
Consider a non-lazy random walk that starts from $u$, visits $v$ and traverse back to $u$.
This whole event must happen at least as late as the first return to $u$, which gives $C_{u,v} \ge \frac{1}{\ppi_u}$ by \eqref{eq:kac-return-time}.
Swapping the order gives $C_{u,v} \ge \frac{1}{\ppi_v}$.
The final lower bound of $\frac{n}{\kappa}$ follows by $\kappa$-uniformity.
\end{proof}

We show below, again, that short random walks suffice for getting high-accuracy estimations of commute time.
\begin{lemma}[Truncation length]
\label{lem:comm_len}
For any $\kappa$-uniform graph $G$ and any $0 < \eps < \frac{1}{2}$, for any truncation length 
$\ell \ge C \kappa \frac{\log(\kappa n \eps^{-1})}{\log n}$, where $C$ is a constant, the commute time satisfies for all distinct vertices $u,v$ that
\[
\Abs{C_{u,v} - \Par{\frac{\hh_{u,v}^{(\ell)}(u)}{2\ppi_u} - \frac{\hh_{u,v}^{(\ell)}(v)}{2\ppi_v}}} \le \eps C_{u,v}
\]
\end{lemma}

\begin{proof}
Since $C_{uv} = \frac{\hh_{u,v}(u)}{2\ppi_u}  - \frac{\hh_{u,v}(v)}{2\ppi_v}$, \Cref{lem:comm_lb} shows that it suffices to prove $|\hh_{u,v}-\hh_{u,v}^{(\ell)}|_\infty \le \eps$.

Using the same argument from \Cref{lem:kappa_mix_len}, we get for every $k \ge \kappa$,
\[
 \left\|\PP^k(\ee_u-\ee_v)\right\|_1
 \leq 2n^{-9\lfloor k/\kappa\rfloor}.
\]
By grouping the tail from $k \ge \ell$ into blocks of $\kappa$ steps, we get
\[
 \left\|\hh_{u,v}-\hh_{u,v}^{(\ell)}\right\|_\infty
 \leq \sum_{k\geq\ell}
       \left\|\PP^k(\ee_u-\ee_v)\right\|_1
 =O\!\left(\kappa n^{-9\lfloor\ell/\kappa\rfloor}\right).
\]
Our choice of $\ell \ge C \kappa \frac{\log (\kappa n \eps^{-1})}{\log n}$, suffices given a sufficiently large constant $C$.
\end{proof}

Similarly, we show that it also suffices to sample for each vertex a small set of random walks for getting a good estimation for the hitting vectors.
\begin{lemma}[Hitting vector concentration]
\label{lem:comm_endpoints}
For any $\kappa$-uniform graph $G$, sample a set of $(\ell,s)$-random walks and let $\widetilde{\hh}_v^{(\ell)}$ be the sampled estimator of $\hh_v^{(\ell)}$.
If $s \ge C\ell\eps^{-2}\log n$, where $C$ is a constant, with high probability, for all pairs of distinct vertices $u,v$,
\[
|\hh_{u,v}^{(\ell)}(u) - \widetilde{\hh}_{u,v}^{(\ell)}(u)| \le \eps.
\]
\end{lemma}
\begin{proof}
We use a similar argument as in the proof of \Cref{lem:stationary-endpoints}.
For vertices $u,x$, $\widetilde{\hh}_u^{(\ell)}(x)$ is the sum over all lengths $0\le k<\ell$ of the averages of $s$ independent terminating vertex indicators at $x$. 
It has expectation exactly $\hh_u^{(\ell)}(x)$.
Each sample has absolute value at most $\frac 1 s$ and the total variance is at most $\frac{\ell}{s}$.
A Bernstein bound then gives
\[
\Pr\left[
\left|\widetilde{\hh}_u^{(\ell)}(x)-\hh_u^{(\ell)}(x)\right|>\eps
\right]
\le 2\exp\!\left(-\Omega(s\eps^2/\ell)\right).
\]
For a sufficiently large constant $C$ in the bound on $s$, a union bound ensures that the inequality holds for all pairs of vertices with high probability. 
Applying it with error $\eps/2$ to $(u,u)$ and $(v,u)$ and using the triangle inequality proves the lemma.
\end{proof}

We have all the ingredients to prove \Cref{thm:dyncomm-formal}.
\begin{proof}[Proof of \cref{thm:dyncomm-formal}]
Our dynamic data structure again uses the dynamic $(\ell,s)$-random walks algorithm (Algorithm~\ref{alg:dyn_rw}) with parameters set to
\[
\ell=\Theta\!\left(
   \kappa\frac{\log(\kappa n\eps^{-1})}{\log n}\right),
   \quad
s=\Theta(\ell\eps^{-2}\log n)
  =\Theta\!\left(\kappa\eps^{-2}
    \log(\kappa n\eps^{-1})\right).
\]
The bound for $\ell$ is taken such that \Cref{lem:comm_len} bounds the truncation error by $\frac{\eps}{8}$ for all entries.
This value of $s$ also meets the sample requirement in \Cref{lem:comm_endpoints} used to estimate the hitting vector differences, giving error $\frac{\eps}{8}$.
By \Cref{lem:comm_lb}, a $(1+\eps/4)$-approximate stationary distribution would then suffice to give a $(1+\eps)$-approximate error for all-pairs commute time.
For such a stationary approximation, we invoke a separate copy of the algorithm from \Cref{thm:dynsta_formal}.

\Cref{lemma:dir_dyn_rw} gives the running time and space bounds for the algorithm.
Note that the time and space bounds in \Cref{thm:dyncomm-formal} are dominated by the current dynamic random walks algorithm.
The query reads 4 counters from the estimated hitting vectors and 2 queries from the approximate stationary, giving $O(1)$ time.
\end{proof}

\section*{Acknowledgement}
AI Disclosure: We used ChatGPT (GPT-5.6 Sol and GPT-6 Astra) to assist with the generation and refinement of mathematical proofs,
a harness written by Junzhao Yang for initial ideas generation,
as well as the Devin.AI harness for verifications.
The tools materially affected Sections 2, 3, 4, and 5.
The authors verified the correctness and originality of all content, including all mathematical arguments and references.

More details on each of the results and their timelines:
\begin{enumerate}
\item The undirected data structure in \cref{sec:undir} was figured out in discussions together with Lawrence Li in Fall 2025,
and has not structurally changed since.
Lawrence was offered authorship but declined,
we thank him for the discussions.
\item The hardness result in \cref{sec:hardness} was initially sketched by the authors, but influenced by a different construction generated by GPT 5.6.
The final simplified construction took components from both ideas, but benefited from bug fixings found using the Devin harness.
\item The directed data structure in \cref{sec:dir} was initially envisioned to be closer to the static algorithms based on Eulerian graphs,
but significantly simplified via GPT 5.6.
\end{enumerate}

\bibliographystyle{alpha}
\bibliography{references}

\newcommand{\etalchar}[1]{$^{#1}$}
\begin{thebibliography}{CGH{\etalchar{+}}20}

\bibitem[BP17]{BringmannP12}
Karl Bringmann and Konstantinos Panagiotou.
\newblock Efficient sampling methods for discrete distributions.
\newblock volume~79, pages 484--508. Springer, 2017.

\bibitem[BPS20]{bernstein2020deterministic}
Aaron Bernstein, Maximilian {Probst Gutenberg}, and Thatchaphol Saranurak.
\newblock Deterministic decremental reachability, scc, and shortest paths via
  directed expanders and congestion balancing.
\newblock In {\em 2020 IEEE 61st Annual Symposium on Foundations of Computer
  Science (FOCS)}, pages 1123--1134, 2020.

\bibitem[CGH{\etalchar{+}}20]{ChenGHPS20}
Li~Chen, Gramoz Goranci, Monika Henzinger, Richard Peng, and Thatchaphol
  Saranurak.
\newblock Fast dynamic cuts, distances and effective resistances via vertex
  sparsifiers.
\newblock In {\em 2020 IEEE 61st Annual Symposium on Foundations of Computer
  Science (FOCS)}, pages 1135--1146, 2020.

\bibitem[Chu14]{C14}
Julia Chuzoy.
\newblock Flows, cuts and integral routing in graphs-an approximation
  algorithmist’s perspective.
\newblock In {\em Proc. of the International Congress of Mathematicians}, pages
  585--607, 2014.

\bibitem[CKL{\etalchar{+}}24]{ChenKLMP24}
Li~Chen, Rasmus Kyng, Yang~P. Liu, Simon Meierhans, and Maximilian
  Probst~Gutenberg.
\newblock Almost-linear time algorithms for incremental graphs: Cycle
  detection, sccs, s-t shortest path, and minimum-cost flow.
\newblock In {\em Proceedings of the 56th Annual ACM Symposium on Theory of
  Computing}, STOC 2024, page 1165–1173, New York, NY, USA, 2024. Association
  for Computing Machinery.

\bibitem[Dev06]{Devroye06}
Luc Devroye.
\newblock Chapter 4 nonuniform random variate generation.
\newblock In Shane~G. Henderson and Barry~L. Nelson, editors, {\em Simulation},
  volume~13 of {\em Handbooks in Operations Research and Management Science},
  pages 83--121. Elsevier, 2006.

\bibitem[EHHL25]{EHL25}
Antoine El-Hayek, Monika Henzinger, and Jason Li.
\newblock Fully dynamic approximate minimum cut in subpolynomial time per
  operation.
\newblock In {\em Proceedings of the 2025 Annual ACM-SIAM Symposium on Discrete
  Algorithms (SODA)}, pages 750--784. SIAM, 2025.

\bibitem[EHHL26]{EHL26}
Antoine El-Hayek, Monika Henzinger, and Jason Li.
\newblock Deterministic and exact fully-dynamic minimum cut of
  superpolylogarithmic size in subpolynomial time.
\newblock In {\em Proceedings of the 2026 Annual ACM-SIAM Symposium on Discrete
  Algorithms (SODA)}, pages 613--663. SIAM, 2026.

\bibitem[GRST21]{GoranciRST21}
Gramoz Goranci, Harald Räcke, Thatchaphol Saranurak, and Zihan Tan.
\newblock The expander hierarchy and its applications to dynamic graph
  algorithms.
\newblock In {\em Proceedings of the 2021 ACM-SIAM Symposium on Discrete
  Algorithms (SODA)}, pages 2212--2228. {SIAM}, 2021.

\bibitem[HKNS15]{HenzingerKNS15}
Monika Henzinger, Sebastian Krinninger, Danupon Nanongkai, and Thatchaphol
  Saranurak.
\newblock Unifying and strengthening hardness for dynamic problems via the
  online matrix-vector multiplication conjecture.
\newblock In {\em Proceedings of the Forty-Seventh Annual ACM Symposium on
  Theory of Computing}, STOC '15, page 21–30, New York, NY, USA, 2015.
  Association for Computing Machinery.

\bibitem[HLS24]{HLS24}
Bernhard Haeupler, Yaowei Long, and Thatchaphol Saranurak.
\newblock Dynamic deterministic constant-approximate distance oracles with
  $n^{\epsilon}$ worst-case update time.
\newblock In {\em 2024 IEEE 65th Annual Symposium on Foundations of Computer
  Science (FOCS)}, pages 2033--2044, 2024.

\bibitem[HPS22]{HenzingerPS22}
Monika Henzinger, Ami Paz, and A.~R. Sricharan.
\newblock {Fine-Grained Complexity Lower Bounds for Families of Dynamic
  Graphs}.
\newblock In Shiri Chechik, Gonzalo Navarro, Eva Rotenberg, and Grzegorz
  Herman, editors, {\em 30th Annual European Symposium on Algorithms (ESA
  2022)}, volume 244 of {\em Leibniz International Proceedings in Informatics
  (LIPIcs)}, pages 65:1--65:14, Dagstuhl, Germany, 2022. Schloss Dagstuhl --
  Leibniz-Zentrum f{\"u}r Informatik.

\bibitem[Knu97]{Knuth97}
Donald~E Knuth.
\newblock {\em The art of computer programming: Seminumerical algorithms,
  volume 2}.
\newblock Addison-Wesley Professional, 1997.

\bibitem[LP17]{LevinP17}
David~A Levin and Yuval Peres.
\newblock {\em Markov chains and mixing times}.
\newblock American Mathematical Society, 2017.

\bibitem[LS23]{LiS23}
Lawrence Li and Sushant Sachdeva.
\newblock A new approach to estimating effective resistances and counting
  spanning trees in expander graphs.
\newblock In {\em Proceedings of the 2023 Annual ACM-SIAM Symposium on Discrete
  Algorithms (SODA)}, pages 2728--2745. SIAM, 2023.

\bibitem[R{\"a}c02]{R02}
Harald R{\"a}cke.
\newblock Minimizing congestion in general networks.
\newblock In {\em The 43rd Annual IEEE Symposium on Foundations of Computer
  Science, 2002. Proceedings.}, pages 43--52, 2002.

\bibitem[ST04]{SpielmanT04}
Daniel~A. Spielman and Shang-Hua Teng.
\newblock Nearly-linear time algorithms for graph partitioning, graph
  sparsification, and solving linear systems.
\newblock In {\em Proceedings of the Thirty-Sixth Annual ACM Symposium on
  Theory of Computing}, STOC '04, page 81–90, New York, NY, USA, 2004.
  Association for Computing Machinery.

\bibitem[SW19]{SaranurakW19}
Thatchaphol Saranurak and Di~Wang.
\newblock {\em Expander Decomposition and Pruning: Faster, Stronger, and
  Simpler}, pages 2616--2635.
\newblock SIAM, 2019.

\end{thebibliography}

\end{document}